 \documentclass[11pt]{article}
 \usepackage{amsfonts}
 \usepackage{mathrsfs}
 \usepackage{bbm}
 \usepackage{amsthm}
\usepackage{amsmath,amssymb}
\usepackage{algorithm}

\theoremstyle{definition}

\theoremstyle{remark}

\theoremstyle{corollary}

\theoremstyle{example}

\theoremstyle{remark}
\newtheorem*{rem}{Remark}

\newtheorem{thm}{Theorem}
\newtheorem{prop}[thm]{Proposition}
\newtheorem{cor}{Corollary}

\usepackage{amssymb}
\usepackage{epsfig}
\usepackage{times}
\title{Randomized Multiple Model Multiple Hypothesis Tracking}

\author{Haiqi Liu, Xiaojing Shen, Zhiguo Wang, Fanqin Meng, Junfeng Wang, Pramod, Varshney
\thanks{This work was supported in part by the NSFC No. 61673282, U1836103 and the PCSIRT16R53.} 
\thanks{Haiqi Liu, Xiaojing Shen (corresponding author), Zhiguo Wang are with School of Mathematics, Sichuan University, Chengdu, Sichuan 610064, China. E-mail: haiqiliu0330@163.com, shenxj@scu.edu.cn, wangzg315@126.com; Junfeng Wang is with the School of Computer Science, Sichuan University, Chengdu, Sichuan 610064, China. E-mail:wangjf@scu.edu.cn.}
\thanks{Fanqin Meng is with School of Aeronautics and Astronautics, Sichuan University, Chengdu, Sichuan 610064, China, also with School of Automation and Information Engineering, Sichuan University of Science and Engineering, Yibin, Sichuan 644000, China. E-mail: mengfanqin2008@163.com.}
\thanks{P. K. Varshney is with the Department of Electrical Engineering and Computer Science, Syracuse University, Syracuse, NY 13244 USA. E-mail: varshney@syr.edu.}
}

\begin{document}
 \maketitle
\begin{abstract}
	This paper considers the data association problem
	for multi-target tracking. Multiple hypothesis tracking is a popular algorithm for solving this problem but it is NP-hard and is is quite complicated for a large number of targets or for tracking maneuvering targets.
	To improve tracking performance and enhance robustness, we propose a randomized multiple model multiple hypothesis tracking method, which has three distinctive advantages.
	First, it yields a randomized data association solution which maximizes the expectation of the logarithm of the posterior probability and can be solved efficiently by linear programming.
	Next, the state estimation performance is improved by the random coefficient matrices Kalman filter, which mitigates the difficulty introduced by randomized data association, i.e., where the coefficient matrices of the dynamic system are random.
	Third, the probability that the target follows a specific dynamic model is derived by jointly optimizing the multiple possible models and data association hypotheses, and it does not require prior mode transition probabilities. Thus, it is more robust for tracking multiple maneuvering targets.
	Simulations demonstrate the efficiency and superior results of the proposed algorithm over interacting multiple model multiple hypothesis tracking.
\end{abstract}

\noindent{\bf keywords:} Multi-target tracking; data association; random coefficient matrices Kalman filter; linear programming

\section{Introduction}
Multi-target tracking (MTT) is an extremely important problem in target tracking  \cite{bar1990multitarget,Vo2015Multitarget}, 
where a complicating factor is the uncertainty of measurement origin namely the data association problem of partitioning measurements into false alarms and true detections \cite{Meyer2018,Zhu2002}. In MTT problems, the number of individual targets and measurements may change over time, since the targets or clutter appear and disappear in the surveillance region.
With multiple targets, the problem becomes even more complicated when there are maneuvering targets leading to the multiple maneuvering target tracking (MMTT) problem. MMTT is much more  difficult than MTT since its challenge is not only the data association problem, but also the uncertainty of target dynamics.

A number of data association algorithms have been developed in the last forty years \cite{Mori2018_ICIF}. One of the earliest data association method is the nearest neighbor (NN) method, which was proposed by Singer et al. \cite{singer1974derivation}. Due to the occurrence of inexact associations in dense clutter environments, some improved algorithms based on NN have been developed, such as the global nearest neighbor (GNN) method \cite{Blackman1986Multiple}. Although GNN considers all possible associations and finds the best association, it exhibits poor performance in MTT scenarios with targets crossing paths.

Probabilistic data association (PDA) \cite{Bar1975Tracking,bar1983consistency} and joint PDA (JPDA) \cite{BarShalom1984JPDA,BarShalom2011Tracking} are probabilistic association approaches to estimate the target state by using the weights based on the Bayesian posterior probability of the valid measurements and target tracks. The difference between PDA and JPDA lies in the way that the association probabilities are calculated. While PDA calculates the association probabilities separately, JPDA computes the joint association probabilities of all the targets.
The computational complexity of the JPDA algorithm increases exponentially with the number of targets and measurements \cite{Vo2015Multitarget}. To handle the issue of high complexity, various approaches have been put forward, such as in \cite{li1999an,TAC2009Oh,williams2014approximate,Luo2012Novel}.

The random finite set (RFS) approach models the multi-target state as an RFS, which propagates the posterior multi-target density \cite{Vo2015Multitarget}.
Although the RFS methodology is a powerful tool, the numerical complexity hinders its practical application.
The probability hypothesis density (PHD) and cardinalized PHD filters have been developed as approximations of the RFS approach, which propagate moments and cardinality distributions recursively \cite{Mahler2003phd,TSP2007VoCPHD,TSP2011MahlerCPHD}.
Other approximations such as the multi-Bernoulli and generalized labeled multi-Bernoulli (GLMB) filters propagate the parameters of a multi-Bernoulli distribution that approximate the posterior multi-target density \cite{TSP2014VoLMBF,hoang2014sensor,TSP2018VoLRFS,TSP2018VoGLMB}.
The sequential Monte Carlo and Gaussian mixture implementations for these filters have been proposed in \cite{Vo2005SMCRFS,TSP2006VoGMPHD,TSP2017VoGLMB}.
The problem of high complexity in these filters still exists, especially for scenarios with a large number of targets.

Multiple hypothesis tracking (MHT) makes a delayed decision regarding measurement-target associations \cite{BarShalom2011Tracking,Blackman2004MHT,Reid1979An,Kuiren1990Issues,Mallick2012MHT}. It seeks to find the highest posterior probability or track score among a set of association hypotheses over a sliding window of consecutive time steps \cite{Vo2015Multitarget,Meyer2018}.
From the viewpoints of track maintenance and cost effectiveness, the track-oriented MHT (TOMHT) \cite{BarShalom2011Tracking,Kuiren1990Issues,Mallick2012MHT} is a popular version of MHT \cite{Vo2015Multitarget}. The TOMHT determines the best global hypothesis by solving an NP-hard multidimensional assignment (MDA) problem \cite{Poore1994Multidimensional}. Its computational complexity increases exponentially with time.
Two approaches can be used to solve the MDA problem, namely the Lagrangian relaxation algorithm \cite{Pattipati2000,Poore1997A} and  approximate linear programming (LP) \cite{Coraluppi2000,Storms2003An}.
Nevertheless, both of them are suboptimal.
The probabilistic MHT (PMHT) is an efficient approach that estimates the target states over a batch of frames via the expectation-maximization (EM) algorithm \cite{PMHT1995,willett2002pmht:}. It reformulates the measurement model of the association problem by removing the requirement that each target is associated with only one measurement. It simplifies computation, but sacrifices estimation performance \cite{Vo2015Multitarget}.

When maneuvering target tracking is considered, multi-model (MM) methods are the most accepted and mainstream approaches to maneuvering target tracking.
Among the variants of the MM method, the interacting multiple model (IMM) method \cite{IMM1988BlomBarShalom,bar2004estimation,li2005survey,BarShalom2011Tracking,ho2011switched} has shown good performance and has been applied in many practical scenarios. The IMM algorithm employs a Markov transition probability matrix to describe the switch between models, and then state estimation is obtained via three steps: interaction, filtering, and combination.
In the framework of PDA and JPDA, interacting multiple model probabilistic data association and interacting multiple model joint probabilistic data association methods have been suggested to handle data association with maneuvering targets \cite{houles1989multisensor,blom2006exact,Bar1991IMM}.
In the framework of MHT, interacting multiple model multiple hypothesis tracking (IMM-MHT) \cite{blackman1999immmht} is a popular method for MMTT.

In this paper, we consider the MMTT problem from the viewpoint of randomized data association, with aims to reduce the computational complexity and improve the association performance and robustness of MMTT.
Our main contributions are summarized as follows:
\begin{itemize}
	\item
	We propose a novel randomized multiple model MHT (RMM-MHT) method by leveraging LP and the random coefficient matrices Kalman filter (RCMKF).
	It has a low computational complexity so that it is suitable for tracking of a large number of targets.
	\item
	The difficulty introduced by randomized data association is that the coefficient matrices of the dynamic system are random, and the traditional RCMKF cannot be employed directly.
	We deal with the difficulty by generalizing the optimal RCMKF to obtain linear minimum variance recursive state estimation.
	\item
	The proposed RMM-MHT method is capable of tracking multiple maneuvering targets.
	It only requires a set of models and does not require the knowledge of prior mode transition probabilities.
	The probabilities of target maneuvering models are derived by jointly optimizing the multiple hypothesis model and data association.	
\end{itemize}

Numerical example demonstrates the superior performance of the proposed algorithm in terms of the optimal sub-pattern assignment (OSPA) distance \cite{SchuhmacherA} and robustness with respect to mode transition probabilities. Our results validate the fact that the proposed RMM-MHT improves the tracking performance.

This paper also extends our earlier work \cite{Yi2017Random} in four aspects:
\begin{itemize}
	\item Proposes a unified framework that includes tracking of maneuvering targets.
	\item Generalizes RCMKF to estimate the state of maneuvering targets.
	\item Provides a detailed computation complexity analysis.
	\item Presents a simulation example that includes both closely spaced targets and maneuvering targets for comparison.
\end{itemize}

The organization of the rest of the paper is as follows: in Section \ref{sec_2}, some background and problem formulation are presented. The RMM-MHT algorithm is formulated in Section \ref{sec_3}. Simulation results are given in Section \ref{sec_4}. We provide our conclusions and possible future work in Section \ref{sec_5}.

\section{Problem Formulation}
\label{sec_2}
\subsection{Multiple Model Multiple Hypothesis Tracking and 0-1 Integer Programming}\label{sec_2_1}
MHT is a deferred decision approach to
data association-based MTT. At each observation time, the
MHT algorithm attempts to propagate and maintain a set
of association hypotheses with the highest track scores \cite{Vo2015Multitarget,Pattipati2000,Poore1997A}.

Assuming that at time $k$, the set of targets is denoted by
 $ \mathbf{T}_k $, i.e.,
\begin{equation}
\label{Eq_2.1} \mathbf{T}_k=\left\{\mathbf{x}^{\tau}_k:~ \tau=0,1,\cdots,T \right\},
\end{equation}
where $ \tau $ is the target index, $ \mathbf{x}^{\tau}_k $ represents the state vector of the target $ \tau $ at time $ k $, and $ T $ is the number of targets (the time index $ k $ is omitted for simplicity when there is no confusion).
A dummy target is introduced by the index $ \tau=0 $, i.e., $ \mathbf{x}^{0}_k $.
An $(N-1)$-scan procedure \cite{Kuiren1990Issues} is used to improve the accuracy of associations between the targets and the measurements, i.e., $N$ scans of measurements are received from time $k+1$ to $k+N$. Let us denote the set of measurements at $k+n$ by
\begin{equation}
\label{Eq_2.2} \mathbf{Z}_{k+n} = \left\{\mathbf{z}_{k+n}^{r_n}:~ r_n=0,1,\cdots,R_n \right\},~n=1,\cdots,N,
\end{equation}
where $ r_n $ is the measurement index, $R_n$ and $\mathbf{z}^{r_n}_{k+n}$ are the number of measurements and the $r_n $-th measurement vector received in the $n$-th scan, respectively.
Similarly, the index $ r_n=0 $ indicates a dummy measurement $ \mathbf{z}_{k+n}^{0} $ in the $ n $-th scan, where $ n=1,\cdots,N $. Notice that the concepts of the dummy target and dummy measurements are to deal with false alarms, missed detections, track initiation and track termination.
Moreover, for tracking of maneuvering targets, a set of dynamic models in each scan is defined as
\begin{equation}\label{Eq_2.3}
 \mathbf{M}_{k+n}=\left\{m_{k+n}^{s_{n}}:~s_{n}=1,\cdots,S_n \right\},~\text{for}~n=1,\cdots,N,
\end{equation}
where $ s_n $ is the model index, $ S_n $ and $ m_{k+n}^{s_n} $ are the number of models and the $ s_n $-th model in the $ n $-th scan, respectively.

In the framework of MHT, a local hypothesis is defined by the indices of a target, measurements and models as follows,
\begin{align}
\nonumber
(\tau,\mathbf{s},\mathbf{r})=(\tau,s_1,\cdots,s_N,r_1,\cdots,r_N),
\end{align}
where $ \tau $ is the target index, $ \mathbf{s}=(s_1,\cdots,s_N) $ is the model index vector that the target $ \tau $ executed in each scan, and  $\mathbf{r}=(r_1,\cdots,r_N)$ is the measurement index vector that is associated with the target $\tau$ in each scan. A global hypothesis is a set of local hypotheses, where each target and each measurement is associated with one local hypothesis.

In the single model MHT problem, the best global hypothesis is determined by maximizing the likelihood or the posterior probability, which is equivalent to a multidimensional assignment problem (MDA) problem \cite{Poore1994Multidimensional}. When we consider multiple models, a $ 0$-$1 $ integer programming (IP) solution to seek the best global hypothesis can be derived as follows.

Define a binary association variable for each local hypothesis, i.e.,
\begin{align}\label{Eq_2.5}
\delta_{(\tau, \mathbf{s}, \mathbf{r})}=\left\{
\begin{array}{rcl}
1 &  &\text{if}~ (\tau,\mathbf{s},\mathbf{r})~\text{is selected}, \\
0 &  &\text{otherwise}.
\end{array}
\right.
\end{align}
A global hypothesis requires the local hypotheses to satisfy the following constraints
\begin{align}\label{Eq_2.6}
&\sum_{\mathbf{s}}\sum_{\mathbf{r}}\delta_{(\tau,\mathbf{s},\mathbf{r})}=1,~\text{for}~\tau=1,\cdots,T,\\
\nonumber
&\sum_{\tau}\sum_{\mathbf{s}}\sum_{\mathbf{r}\setminus\{r_n\}}\delta_{(\tau,\mathbf{s},\mathbf{r})}=1,\\
\label{Eq_2.61}
&~\text{for}~r_n=1,\cdots,R_n~\text{and}~n=1,\cdots,N,
\end{align}
where
\begin{align}\nonumber
\sum\limits_{\mathbf{r}\setminus\{r_n\}}\{\cdot\}=\sum_{r_1}\cdots\sum_{r_{n-1}}\sum_{r_{n+1}}\cdots\sum_{r_N} \{\cdot\}.
\end{align}
Equations (\ref{Eq_2.6})-(\ref{Eq_2.61}) mean that each target and each measurement can be associated with only one local hypothesis, respectively. For each local hypothesis $ (\tau,\mathbf{s},\mathbf{r}) $, the track score or the posterior probability is calculated to evaluate its quality, which is denoted by $ L_{(\tau,\mathbf{s},\mathbf{r})} $. The specific form of $ L_{(\tau,\mathbf{s},\mathbf{r})} $ is given in Appendix A.
Thus, the objective function of multiple-model MHT is expressed in a product form as follows:
\begin{align}\label{Eq_2.7}
\max_{\delta_{(\tau,\mathbf{s},\mathbf{r})}}~ \prod_{\tau}\prod_{\mathbf{s}}\prod_{\mathbf{r}}\left[L_{(\tau,\mathbf{s},\mathbf{r})}\right]^{\delta_{(\tau,\mathbf{s},\mathbf{r})}}.
\end{align}
By taking the negative logarithm of the above objective function (\ref{Eq_2.7}), a linear objective function is obtained, i.e.,
\begin{align}\nonumber
\min_{\delta_{(\tau,\mathbf{s},\mathbf{r})}}~  \sum_{\tau}\sum_{\mathbf{s}}\sum_{\mathbf{r}}
\left[ -\log L_{(\tau,\mathbf{s},\mathbf{r})}\right] {\delta_{(\tau,\mathbf{s},\mathbf{r})}},
\end{align}
where $ -\log L_{(\tau,\mathbf{s},\mathbf{r})} $ is defined as the cost of the track $ (\tau,\mathbf{s},\mathbf{r}) $.
Letting
\begin{align}\label{cost}
C_{(\tau,\mathbf{s},\mathbf{r})}=-\log L_{(\tau,\mathbf{s},\mathbf{r})},
\end{align}
it can be easily seen that the multiple-model MHT problem is equivalent to a $ 0 $-$ 1 $ IP:
\begin{align}\label{Eq_2.8}
\nonumber \text{Minimize}~
&\sum_{\tau}\sum_{\mathbf{s}}\sum_{\mathbf{r}}C_{(\tau,\mathbf{s},\mathbf{r})}\delta_{(\tau,\mathbf{s},\mathbf{r})}\\
\nonumber
\text{Subject to:}\\
\nonumber
&\sum_{\mathbf{s}}\sum_{\mathbf{r}}\delta_{(\tau,\mathbf{s},\mathbf{r})}=1, ~\text{for}~ \tau=1,\cdots,T,\\
\nonumber
&\sum_{\tau}\sum_{\mathbf{s}}\sum_{\mathbf{r}\setminus\{r_n\}}\delta_{(\tau,\mathbf{s},\mathbf{r})}=1,~\\
\nonumber
&\text{for}~r_n=1,\cdots,R_n~\text{and}~n=1,\cdots,N,\\
&\delta_{(\tau,\mathbf{s},\mathbf{r})} \in\{0,1\}, ~\text{for all}~ (\tau,\mathbf{s},\mathbf{r}).
\end{align}
The $0$-$1$ IP problem $(\ref{Eq_2.8})$ is an NP-hard problem whose computational complexity grows exponentially with the problem size, and it reduces to an MDA problem when the model set only includes one model, i.e., $S_n=1$, for $ n=1,\cdots,N $.

\subsection{Random Coefficient Matrices Kalman Filters}\label{sec_2_2}
In order to deal with the MTT problem, we introduce the RCMKF \cite{Koning1984Optimal}. Consider a discrete time dynamic system
\begin{align}
\label{Eq_2.9} &\mathbf{x}_{k+1} = \mathbf{F}_k\mathbf{x}_k+\mathbf{v}_k \\
\label{Eq_2.10} &\mathbf{z}_{k+1} = \mathbf{H}_{k+1}\mathbf{x}_{k+1}+\mathbf{w}_{k+1},
\end{align}
where $k$ is the time index, $\mathbf{x}_k$ is the system state, $\mathbf{z}_k$ is the measurement. $\mathbf{F}_k$ and $\mathbf{v}_k$ are the state transition matrix and the process noise, respectively. $\mathbf{H}_k$ and $\mathbf{w}_k$ are the measurement matrix and the measurement noise, respectively.
Here, we assume that the state equation (\ref{Eq_2.9}) is random with a discrete distribution, which has $S$ realizations, i.e.,
\begin{align}\label{Eq_2.11}
\nonumber  \mathbf{x}_{k+1} &= F_k^{(1)}\mathbf{x}_k+v_k^{(1)}~\text{with probability} ~~P_{1}^f\\
&=\cdots\\
\nonumber &= F_k^{(S)}\mathbf{x}_k+v_k^{(S)}~\text{with probability} ~~P_{S}^f,
\end{align}
where $F_k^{(s)}$ is a real matrix, which is the $s$-th realization of $\mathbf{F}_k$, and $v_k^{(s)}$ is the process noise, which is a random vector, for $ s=1,\cdots,S $.
Similarly, the measurement equation (\ref{Eq_2.10}) follows a discrete distribution with $M$ realizations, i.e.,
\begin{align}\label{Eq_2.12}
\mathbf{z}_{k+1}
\nonumber&= H_{k+1}^{(1)}\mathbf{x}_{k+1} +  w_{k+1}^{(1)} ~\text{with probability} ~~P_{1}^h\\
&= \cdots \\
\nonumber&= H_{k+1}^{(M)}\mathbf{x}_{k+1}+w_{k+1}^{(M)} ~\text{with probability} ~~P_{M}^h,
\end{align}
where $ H_{k+1}^{(m)} $ is a real matrix, which is the $m$-th realization of $ \mathbf{H}_{k+1} $, and $ w_{k+1}^{(m)} $ is the measurement noise, for $ m=1,\cdots,M $.
Notice that the state transition matrix $ \mathbf{F}_k $ and measurement matrix $ \mathbf{H}_{k+1} $ defined in $(\ref{Eq_2.11})$ and $(\ref{Eq_2.12})$ are random coefficient matrices, which are  different from those in the standard Kalman filter. Next, we consider the filtering problem for handling the dynamic system $(\ref{Eq_2.11})$ and $(\ref{Eq_2.12})$.

To derive the linear minimum variance recursive estimation formulation, it is required that the system satisfies the following conditions:
\begin{enumerate}
	\item $\{\mathbf{F}_k,\mathbf{H}_k,\{v^{(s)}_k\}_{s=1}^{S},\{w_{k}^{(m)}\}_{m=1}^M,k=0,1,2,\cdots\}$ are temporal sequences of independent random variables and the initial state $\mathbf{x}_0$ is independent of them.
	\item $\mathbf{\mathbf{x}_k}$ and $\{\mathbf{F}_k,\mathbf{H}_k, k=0,1,2,\cdots\}$ are mutually independent.
	\item The initial state $\mathbf{x}_0$, the noises $ \{v^{(s)}_k\}_{s=1}^{S}$ and $\{w_{k}^{(m)}\}_{m=1}^M $ have the following means and covariances,
	\begin{align}
	\nonumber
	&E(\mathbf{x}_0)=\mu_0,~~~ E(\mathbf{x}_0-\mu_0)(\mathbf{x}_0-\mu_0)'=\mathbf{P}_{0|0},  \\
	\nonumber
	&E(v_k^{(s)})=0,~~~ E(v^{(s)}_k(v^{(s)}_k)')=\mathbf{Q}_k^{(s)}, \\
	\nonumber
	&E(w^{(m)}_{k+1})=0,~~ E(w^{(m)}_{k+1}(w^{(m)}_{k+1})')=\mathbf{R}_{k+1}^{(m)},
	\end{align}
	where $ s=1,\cdots,S $ and $ m=1,\cdots,M $.
\end{enumerate}

We denote the following as
\begin{align}
\label{Eq_2.13} \widetilde{\mathbf{F}}_k &= \mathbf{F}_k-\overline{\mathbf{F}},\\
\label{Eq_2.14} \widetilde{\mathbf{H}}_{k+1} &= \mathbf{H}_{k+1}-\overline{\mathbf{H}},
\end{align}
where $\overline{\mathbf{F}}=E(\mathbf{F}_k)$ and $\overline{\mathbf{H}}=E(\mathbf{H}_{k+1})$.
Substituting $(\ref{Eq_2.13})$ and $(\ref{Eq_2.14})$ into the dynamic system $(\ref{Eq_2.9})$ and $(\ref{Eq_2.10})$, we have
\begin{align}
\label{Eq_2.15} \mathbf{x}_{k+1} &= \overline{\mathbf{F}}\mathbf{x}_k+\widetilde{\mathbf{v}}_k \\
\label{Eq_2.16} \mathbf{z}_{k+1} &= \overline{\mathbf{H}}\mathbf{x}_{k+1}+\widetilde{\mathbf{w}}_{k+1},
\end{align}
where $ \widetilde{\mathbf{v}}_k = \mathbf{v}_k+\widetilde{\mathbf{F}}_k\mathbf{x}_k $ and $ \widetilde{\mathbf{w}}_{k+1} = \mathbf{w}_k+\widetilde{\mathbf{H}}_k\mathbf{x}_{k+1} $.
Note that the noise process $\widetilde{\mathbf{v}}_k$ and $\widetilde{\mathbf{w}}_{k+1}$ are zero-mean, and their covariances are denoted by $\widetilde{\mathbf{Q}}_k$ and $\widetilde{\mathbf{R}}_{k+1}$, respectively.

De Koning \cite{Koning1984Optimal} proposed the linear minimum variance recursive estimation formulation for the linear discrete time dynamic system with random state transition and measurement matrices. Such recursive estimation formulation only considers the case of identical process and measurement noises, which is not suitable for the problem stated in $(\ref{Eq_2.11})$ and $(\ref{Eq_2.12})$.
To this end, the recursive estimation formulation of the new system $(\ref{Eq_2.15})$ and $(\ref{Eq_2.16})$ is generalized as follows:
\begin{prop}\label{pro_1}
	The linear minimum variance recursive state estimation of system $(\ref{Eq_2.15})$ and $(\ref{Eq_2.16})$ is given by
	\begin{align}
	\label{RCKF_1}
	&\mathbf{x}_{k+1|k+1} = \mathbf{x}_{k+1|k}+\mathbf{K}_{k+1}(\mathbf{z}_{k+1}- \overline{\mathbf{H}} \mathbf{x}_{k+1|k})\\
	\label{RCKF_2}
	&\mathbf{x}_{k+1|k} = \overline{\mathbf{F}} \mathbf{x}_{k|k} \\
	\label{RCKF_3}
	&\mathbf{P}_{k+1|k} = \overline{\mathbf{F}} \mathbf{P}_{k|k} \overline{\mathbf{F}}' +\widetilde{\mathbf{Q}}_k\\
	\label{RCKF_4}
	&\mathbf{K}_{k+1} = \mathbf{P}_{k+1|k} \overline{\mathbf{H}}' (\overline{\mathbf{H}}
    \mathbf{P}_{k+1|k} \overline{\mathbf{H}}' +\widetilde{\mathbf{R}}_{k+1})^\dagger \\
	&\label{RCKF_5}   \mathbf{P}_{k+1|k+1} = (\mathbf{I}-\mathbf{K}_{k+1}\overline{\mathbf{H}})\mathbf{P}_{k+1|k}\\
	\label{RCKF_6}
	&\widetilde{\mathbf{Q}}_k = \sum\nolimits_{s} P_s^f \left\{ \mathbf{Q}_k^{(s)} +(F^{(s)}_k-\overline{\mathbf{F}}) E(\mathbf{x}_k\mathbf{x}_k') (F^{(s)}_k-\overline{\mathbf{F}})'\right\}\\
	\label{RCKF_7}
	&\widetilde{\mathbf{R}}_{k+1} = \sum\nolimits_{m}P_m^h \left\{\mathbf{R}_{k+1}^{(m)} +(H^{(m)}_{k+1}-\overline{\mathbf{H}}) E(\mathbf{x}_{k+1}\mathbf{x}_{k+1}') (H^{(m)}_{k+1}-\overline{\mathbf{H}})'\right\}\\
	\label{RCKF_8}
	&E(\mathbf{x}_{k+1}\mathbf{x}_{k+1}') = \overline{\mathbf{F}} E(\mathbf{x}_k\mathbf{x}_k') \overline{\mathbf{F}}' + \widetilde{\mathbf{Q}}_k \\
	\nonumber
	&\mathbf{x}_{0|0} = E\mathbf{x}_0,~\mathbf{P}_{0|0}=\text{Cov}(\mathbf{x}_0),\\
    &E(\mathbf{x}_0\mathbf{x}_0') = E\mathbf{x}_0E\mathbf{x}_0'+\mathbf{P}_{0|0}.
	\end{align}
\end{prop}
Notice that the generalized RCMKF in Proposition \ref{pro_1} is a Kalman type filter. Equations (\ref{RCKF_1})-(\ref{RCKF_5}) are similar to the standard Kalman filter. The main difference is the calculations of $ \widetilde{\mathbf{Q}}_k $ and $ \widetilde{\mathbf{R}}_{k+1} $, i.e.,  (\ref{RCKF_6})-(\ref{RCKF_7}), which result from the randomness of $ \mathbf{F}_k $ and $ \mathbf{H}_{k+1} $. Both of them require the recursive calculation of $ E(\mathbf{x}_{k}\mathbf{x}_{k}') $ by (\ref{RCKF_8}).

\section{Randomized Multiple Model MHT Algorithm}\label{sec_3}
\subsection{Randomized Data Association Decision}
In this section, we consider the data association problem from the point of view of randomized association decisions.
Define the probability of the local hypothesis $ (\tau,\mathbf{s},\mathbf{r}) $ by $ P_{(\tau,\mathbf{s},\mathbf{r})} $, which satisfies the following constraints
\begin{align}\label{Eq_3.18}
&\sum_{\mathbf{s}}\sum_{\mathbf{r}}P_{(\tau,\mathbf{s},\mathbf{r})}=1,~\text{for}~\tau=1,\cdots,T,\\
\nonumber
&\sum_{\tau}\sum_{\mathbf{s}}\sum_{\mathbf{r}\setminus\{r_n\}}P_{(\tau,\mathbf{s},\mathbf{r})}=1,\\
\label{Eq_3.181}
&\text{for}~r_n=1,\cdots,R_n~\text{and}~n=1,\cdots,N,\\
\label{Eq_3.182}
&0\leq P_{(\tau,\mathbf{s},\mathbf{r})} \leq 1~\text{for all}~(\tau,\mathbf{s},\mathbf{r}).
\end{align}
Equations (\ref{Eq_3.18})-(\ref{Eq_3.182}) mean that the sum of the probabilities of each target and each measurement being associated with possibly multiple possible local hypotheses is equal to 1. Different from the deterministic binary association variable (\ref{Eq_2.5}), the randomized decision variable is defined in terms of the probability $P_{(\tau,\mathbf{s},\mathbf{r})}$ as follows,
\begin{align}\label{Eq_3.17}
\delta_{(\tau,\mathbf{s},\mathbf{r})}= \left\{
\begin{array}{rcl}
1& &\text{with probability}~ P_{(\tau,\mathbf{s},\mathbf{r})} \\
0& &\text{with probability}~ 1- P_{(\tau,\mathbf{s},\mathbf{r})},
\end{array}\right.
\end{align}
which selects the local hypothesis $ (\tau,\mathbf{s},\mathbf{r})$ with probability $P_{(\tau,\mathbf{s},\mathbf{r})}$.

For randomized decisions, one may maximize the posterior probability or the logarithm of the posterior probability or minimize the maximum loss or minimize the average loss (Bayes risk). Here, we maximize the expectation of the logarithm of the posterior probability (\ref{Eq_2.7}) with respect to the randomized decisions $\delta_{(\tau,\mathbf{s},\mathbf{r})}$, i.e.,
\begin{align}\label{Eq_3.19}
\nonumber
&\max_{\delta_{(\tau,\mathbf{s},\mathbf{r})}}~ E_{\delta_{(\tau,\mathbf{s},\mathbf{r})}}\left(\log\left[\prod_{\tau}\prod_{\mathbf{s}}\prod_{\mathbf{r}} \left[ L_{(\tau,\mathbf{s},\mathbf{r})}\right] ^{\delta_{(\tau,\mathbf{s},\mathbf{r})}}\right]\right)\\
=~&\min_{\delta_{(\tau,\mathbf{s},\mathbf{r})}}~ E_{\delta_{(\tau,\mathbf{s},\mathbf{r})}}\left(\sum_{\tau}\sum_{\mathbf{s}}\sum_{\mathbf{r}}\left[-\log L_{(\tau,\mathbf{s},\mathbf{r})} \right]{\delta_{(\tau,\mathbf{s},\mathbf{r})}}\right).
\end{align}
From (\ref{Eq_3.17}), we have $E(\delta_{(\tau,\mathbf{s},\mathbf{r})})= P_{(\tau,\mathbf{s},\mathbf{r})}$. Thus, the objective function (\ref{Eq_3.19}) is equivalent to
\begin{align}
\nonumber
\min_{P_{(\tau,\mathbf{s},\mathbf{r})}}~ \sum_{\tau}\sum_{\mathbf{s}}\sum_{\mathbf{r}}\left[-\log L_{(\tau,\mathbf{s},\mathbf{r})} \right]{P_{(\tau,\mathbf{s},\mathbf{r})}}.
\end{align}
Similarly, with the constraints (\ref{Eq_3.18})-(\ref{Eq_3.181}) and the cost coefficient (\ref{cost}), the RMM-MHT problem is converted into a linear programming (LP) formulation as follows:
\begin{align}\label{Eq_3.21}
\nonumber \min~
&\sum_{\tau}\sum_{\mathbf{s}}\sum_{\mathbf{r}}C_{(\tau,\mathbf{s},\mathbf{r})} P_{(\tau,\mathbf{s},\mathbf{r})}\\
\nonumber
\text{Subject to:}\\
\nonumber
&\sum_{\mathbf{s}}\sum_{\mathbf{r}}P_{(\tau,\mathbf{s},\mathbf{r})}=1,~\text{for}~\tau=1,\cdots,T,\\
\nonumber
&\sum_{\tau}\sum_{\mathbf{s}}\sum_{\mathbf{r}\setminus\{r_n\}}P_{(\tau,\mathbf{s},\mathbf{r})}=1,\\
\nonumber
&~\text{for}~r_n=1,\cdots,R_n~\text{and}~n=1,\cdots,N,\\
& ~0\leq P_{(\tau,\mathbf{s},\mathbf{r})} \leq 1~\text{for all}~(\tau,\mathbf{s},\mathbf{r}).
\end{align}
Notice that this LP problem (\ref{Eq_3.21}) is the relaxed version of problem (\ref{Eq_2.8}), which replaces the nonconvex $\{0,1\}$ constraints by the convex interval $[0,1]$ constraints.
It can be solved efficiently. Moreover, the performance of RMM-MHT in the data association step can be guaranteed by the following result.

\begin{prop}
    Under the criterion of maximizing the expectation of the logarithm of the posterior probability, the optimal objective function value of the optimization problem (\ref{Eq_3.21}) based on randomized association decision is not worse than that of the optimization problem (\ref{Eq_2.8}) based on deterministic association decision.
\end{prop}
\begin{proof}
	Since both the objective function and the constraint set of the optimization problem (\ref{Eq_2.8}) based on deterministic association decision are special cases of the optimization problem (\ref{Eq_3.21}) based on randomized association decision by setting $ P_{(\tau,\mathbf{s},\mathbf{r})}=0/1 $, the optimal objective function value of the problem (\ref{Eq_3.21}) is not worse than that of the problem (\ref{Eq_2.8}) under the criterion of maximizing the expectation of the logarithm of the posterior probability.
\end{proof}

\begin{rem}
	The definition of probabilities in RMM-MHT is different from that for PMHT \cite{PMHT1995}.
	RMM-MHT defines the probability for each local hypothesis (see (\ref{Eq_3.17})) and can be obtained by solving an LP (\ref{Eq_3.21}), while PMHT defines the probability for measurement/target association and obtains the state estimation by the EM algorithm with posterior probability \cite{PMHT1995,willett2002pmht:}.
	The PMHT does not employ the feasible association event constraint that a measurement can have only one source and at most one measurement can originate from a target \cite{Vo2015Multitarget,willett2002pmht:}. Also, it assumes the measurement/target association process as independent across measurements.
	By doing so, it is able to render a fully optimal (under the modified assumption) tracker \cite{PMHT1995,willett2002pmht:}.
\end{rem}

Note that, the difficulty introduced by randomized data association is that the state and measurement equations are random. We deal with the difficulty by the optimal RCMKF in Proposition 1.

\subsection{Random Measurement Equation}
The optimal solution of the problem $(\ref{Eq_3.21})$, i.e., randomized data association, yields the association probabilities of measurement to target assignment. Specifically, let $P_{\tau,r_1}^h$ denote the association probability that the measurement $z_{k+1}^{r_1}$ is associated with the target $\tau$, i.e.,
\begin{equation}\label{Eq_3.22}
P_{\tau,r_1}^h=\sum_{\mathbf{s}}\sum_{(r_2,\cdots,r_n)}P_{(\tau,\mathbf{s},\mathbf{r})}.
\end{equation}
Then, based on the association probability (\ref{Eq_3.22}), we define a random coefficient matrices measurement equation,
\begin{align*}
	\mathbf{z}_{k+1}^{r_1} = \mathbf{h}_{k+1}^{r_1} \mathbf{X}_{k+1} + w_{k+1}^{r_1},
\end{align*}
where
\begin{align*}
	\mathbf{X}_{k+1} = \left((\mathbf{x}_{k+1}^1)',\cdots,(\mathbf{x}_{k+1}^T)'\right)'
\end{align*}
is the stacked state vector, and
\begin{align*}
	\mathbf{h}_{k+1}^{r_1} &= \left(H_{k+1}',0,\cdots,0\right)',\quad \text{with probability $P_{1,r_1}^h$},\\
		&=\cdots\\
	&=\left(0,0,\cdots,H_{k+1}'\right)',\quad \text{with probability $P_{T,r_1}^h$}.
\end{align*}
is a random coefficient matrix. Moreover, the random coefficient matrix measurement equation for the stacked measurement vector $\mathbf{z}_{k+1}$ is
\begin{align}\label{Eq_3.24}
	\mathbf{z}_{k+1} = \mathbf{h}_{k+1}\mathbf{X}_{k+1}+\mathbf{w}_{k+1},
\end{align}
where
\begin{align}
	\nonumber&\mathbf{z}_{k+1} = \left((\mathbf{z}_{k+1}^{1})',\cdots,(\mathbf{z}_{k+1}^{R_1})'\right)',\\
	\nonumber&\mathbf{h}_{k+1} = \left( (\mathbf{h}_{k+1}^{1})',\cdots,(\mathbf{h}_{k+1}^{R_1})' \right),\\
	\nonumber&\mathbf{w}_{k+1} = \left((\mathbf{w}_{k+1}^1)',\cdots,(\mathbf{w}_{k+1}^{R_1})'\right)',\\
	\nonumber&\text{Cov}(\mathbf{w}_{k+1}) = \mathbf{R}_{k+1}.
\end{align}

\subsection{Random State Equation}
Similarly, the probability of the target $\tau$ following the $s_1$-th model derived from the optimal solution of problem $(\ref{Eq_3.21})$ is obtained by
\begin{align}\label{Eq_3.27}
	P_{\tau,s_1}^f=\sum_{(s_2,\cdots,s_N)}\sum_{\mathbf{r}}P_{(\tau,\mathbf{s},\mathbf{r})}.
	\end{align}
	Note that it is different from the IMM algorithm, the  probability $(\ref{Eq_3.27})$ comes from the optimal solution of problem $(\ref{Eq_3.21})$ by jointly optimizing the multiple possible models and data association hypotheses.
	With the probability $(\ref{Eq_3.27})$, the random state equation for each target is defined as follows,
	\begin{align}\label{Eq_3.29}
	\mathbf{x}_{k+1}^{\tau} &= \mathbf{F}^{\tau}_k\mathbf{x}_k^{\tau}+\mathbf{v}_k^{\tau},
	\end{align}
	where $ \tau=1,\cdots,T $, and
	\begin{align}
	\nonumber  \mathbf{x}_{k+1}^{\tau} &= F_k^{(1)}\mathbf{x}_k^{\tau}+v_k^{(1)}~\text{with probability} ~~P_{\tau,1}^f\\
	\nonumber  &=\cdots\\
	\nonumber &= F_k^{(S_1)}\mathbf{x}_k^{\tau}+v_k^{(S_1)}~\text{with probability} ~~P_{\tau,S_1}^f,
	\end{align}
	where $F_k^{(s_1)}$ is the state transition matrix, and $v_k^{(s_1)}$ is the process noise with zero-mean and covariance $\mathbf{Q}^{(s_1)}_k$, for $ s_1=1,\cdots,S_1 $.
	Thus, the random state equation for the stacked state is defined as
	\begin{equation}
	\label{Eq_3.30} \mathbf{X}_{k+1} = \mathbf{F}_k\mathbf{X}_k+\mathbf{v}_k,
	\end{equation}
	where
	\begin{align}
	\nonumber &\mathbf{F}_k = \text{diag}\left( \mathbf{F}_k^{1},\cdots,\mathbf{F}_k^{T}\right) , \\
	\nonumber
	&\mathbf{v}_k = \left( (\mathbf{v}_k^{1})',\cdots,(\mathbf{v}_k^{T})'\right),\\
	\nonumber &\text{Cov}(\mathbf{v}_k) = \mathbf{Q}_k.
	\end{align}
	
	\subsection{State Estimation}
	The randomized data association method results in the state equation $(\ref{Eq_3.30})$ and the measurement equation $(\ref{Eq_3.24})$ with random coefficient matrices. Again, similar to $ (\ref{Eq_2.13}) $ and $ (\ref{Eq_2.14}) $, we have
	\begin{align}
	\nonumber  &\widetilde{\mathbf{F}}_{k} = \mathbf{F}_{k}-\overline{\mathbf{F}},\\
	\nonumber &\tilde{\mathbf{h}}_{k+1} = \mathbf{h}_{k+1}-\bar{\mathbf{h}},
	\end{align}
	where $\overline{\mathbf{F}} = E(\mathbf{F}_k)$ and  $\bar{\mathbf{h}} = E(\mathbf{h}_{k+1})$. Thus, $(\ref{Eq_3.30})$ and $(\ref{Eq_3.24})$ are rewritten as
	\begin{align}
	\label{Eq_3.31}
	\mathbf{X}_{k+1} &= \overline{\mathbf{F}}\mathbf{X}_k+\tilde{\mathbf{v}}_{k},\\
	\label{Eq_3.32}
	\mathbf{z}_{k+1} &= \bar{\mathbf{h}}\mathbf{X}_{k+1}+\tilde{\mathbf{w}}_{k+1},
	\end{align}
	where
	\begin{align}
	\nonumber &\tilde{\mathbf{v}}_{k}=\mathbf{v}_{k}+\widetilde{\mathbf{F}}_{k}\mathbf{X}_{k},\\
	\nonumber &\tilde{\mathbf{w}}_{k+1}= \mathbf{w}_{k+1}+\tilde{\mathbf{h}}_{k+1}\mathbf{X}_{k+1}.
	\end{align}
	By Proposition \ref{pro_1}, the linear minimum variance recursive state estimate of the stacked state vector $\mathbf{X}_{k+1}$ is given by
	\begin{equation}
	\label{Eq_3.33}\mathbf{X}_{k+1|k+1} = \mathbf{X}_{k+1|k}+\mathbf{K}_{k+1}(\mathbf{z}_{k+1}-\bar{\mathbf{h}}\mathbf{X}_{k+1|k}),
	\end{equation}
	where
	\begin{align}
	\label{Eq_3.34}
	&\mathbf{X}_{k+1|k} = \overline{\mathbf{F}}\mathbf{X}_{k|k},\\
	\label{Eq_3.35}
	&\widetilde{\mathbf{Q}}_k = \mathbf{Q}_k+E(\widetilde{\mathbf{F}}_kE(\mathbf{X}_k\mathbf{X}_k')\widetilde{\mathbf{F}}_k'),\\
	\label{Eq_3.36}
	&\mathbf{P}_{k+1|k} = \overline{\mathbf{F}} \mathbf{P}_{k|k} \overline{\mathbf{F}}' +\widetilde{\mathbf{Q}}_k,\\
	\label{Eq_3.37}
	&E(\mathbf{X}_{k+1}\mathbf{X}_{k+1}') = \overline{\mathbf{F}} E(\mathbf{X}_{k}\mathbf{X}_{k}') \overline{\mathbf{F}}' +E(\widetilde{\mathbf{F}}_kE(\mathbf{X}_{k}\mathbf{X}_{k}')\widetilde{\mathbf{F}}_k')+\mathbf{Q}_k,\\
	\label{Eq_3.38}
	&\widetilde{\mathbf{R}}_{k+1} = \mathbf{R}_{k+1}+E(\tilde{\mathbf{h}}E(\mathbf{X}_{k+1}\mathbf{X}_{k+1}')\tilde{\mathbf{h}}'),\\
	\label{Eq_3.39}
	&\mathbf{K}_{k+1} = \mathbf{P}_{k+1|k}\bar{\mathbf{h}}'(\bar{\mathbf{h}}
	\mathbf{P}_{k+1|k}\bar{\mathbf{h}}'+\widetilde{\mathbf{R}}_{k+1})^\dagger,\\
	\label{Eq_3.40}
	&\mathbf{P}_{k+1|k+1} = (\mathbf{I}-\mathbf{K}_{k+1}\bar{\mathbf{h}})\mathbf{P}_{k+1|k}.
	\end{align}
	\begin{rem}
		Once the targets begin to share measurements, the estimate of the stacked state is considered as a correlated state estimate.
		It results in the fact that the estimation error covariance $ \mathbf{P}_{k+1|k+1} $ is not block diagonal. Thus, the correlation between the states of multiple targets affects the cross covariance terms of the estimation error covariance of the optimal state estimation RCMKF.
		Besides, the analytical expressions of $ \widetilde{\mathbf{Q}}_k $ and $ \widetilde{\mathbf{R}}_{k+1} $ for the target $\tau$ are given by
		\begin{align}
		\nonumber
		&(\widetilde{\mathbf{Q}}_k)_{(\tau)}=\sum_{s_1=1}^{S_1}P_{\tau,s_1}^f \cdot \left[\mathbf{Q}^{(s_1)}_k+ (F^{(s_1)}_{k}-\overline{\mathbf{F}})\mathbf{X}_{\tau}^{k} (F^{(s_1)}_{k}-\overline{\mathbf{F}})'\right],\\
		\nonumber
		& (\widetilde{\mathbf{R}}_{k+1})_{(\tau)}=\sum_{r_1=1}^{M_{\tau}} P_{\tau,r_1}^h \cdot\left[\mathbf{R}^{(\tau)}_{k+1} + (h^{\tau,r_1}_{k+1}-\bar{\mathbf{h}}_{k+1}^{\tau}) E(\mathbf{X}_{k+1}\mathbf{X}_{k+1}') (h^{\tau,r_1}_{k+1}-\bar{\mathbf{h}}_{k+1}^{\tau})'\right],
		\end{align}
		where $\tau = 1,\cdots,T$, and $\mathbf{X}^k_{\tau}$ is the $\tau$-th diagonal block of $\mathbf{X}^k = E(\mathbf{X}_{k}\mathbf{X}_{k}')$, and $ F^{(s_1)}_{k} $ and $ h^{\tau,r_1}_{k+1} $ are the $ s_1 $-th and $ r_1 $-th realizations of $ \mathbf{F}^{\tau}_{k} $ and $ \mathbf{h}^{\tau}_{k+1} $, respectively.
	\end{rem}
	
	Therefore, the tracking result of each target is
	\begin{equation}
	\nonumber \mathbf{x}_{k+1|k+1}^{\tau} = \mathbf{X}_{k+1|k+1}(*\tau),~~\tau = 1,\cdots,T,
	\end{equation}
	where $\mathbf{X}_{k+1|k+1}(*\tau)$ is the $\tau$-th block of vector $\mathbf{X}_{k+1|k+1}$.
	
	Note that once the association probability is given, $(\ref{Eq_3.33})$-$(\ref{Eq_3.40})$ represent the linear minimum variance recursive state estimation method since the system $(\ref{Eq_3.31})$ and $(\ref{Eq_3.32})$ satisfy the conditions of Proposition \ref{pro_1}.
	
	\begin{rem}
		The major computationally intensive parts in the RMM-MHT algorithm are the LP problem $(\ref{Eq_3.21})$ and the matrix inversion of Equation $(\ref{Eq_3.39})$ with the worst-case computational complexity of $O((R_1\cdot m)^3)$, where $R_1$ is the number of valid measurements, and $m$ is the dimension of the measurement.
		In general, the computational complexity of the LP algorithm is much less than that of the Lagrangian relaxation algorithm in MHT.
	\end{rem}
	
	Furthermore, since the optimal solution of LP may not be in the form of $\{0,1\}$, one may be curious about the conditions that make the solutions of both problems $(\ref{Eq_2.8})$ and $(\ref{Eq_3.21})$ equivalent.
	\begin{prop}\label{pro_2}
		Let $\delta$ be a feasible solution of the $0$-$1$ IP problem $(\ref{Eq_2.8})$, and for $\tau=1,\cdots,T$, $\delta_{(\tau,\mathbf{s}^{\tau},\mathbf{r}^{\tau})}=1$, where $ (\tau,\mathbf{s}^\tau,\mathbf{r}^\tau)=(\tau,s_1^\tau,\cdots,s_N^\tau,r_1^\tau,\cdots,r_N^\tau) $.
		If the cost of $\delta$ satisfies:
		\begin{equation}\label{Eq_3.41}
			C_{(\tau,\mathbf{s}^{\tau},\mathbf{r}^{\tau})}=\min_{(\mathbf{s},\mathbf{r})}\{C_{(\tau,\mathbf{s},\mathbf{r})}\}
		\end{equation}
		where $\tau=1,\cdots,T$, and
		\begin{equation}\label{Eq_3.42}
		C_{(0,\mathbf{s},\mathbf{r})}=0,~\text{for all}~(\mathbf{s},\mathbf{r}),
		\end{equation}
		then the optimal objective function value of the LP problem $(\ref{Eq_3.21})$ is equivalent to that of the $0$-$1$ IP problem $(\ref{Eq_2.8})$, and $\delta$ is the optimal solution of both problems $(\ref{Eq_2.8})$ and $(\ref{Eq_3.21})$.
	\end{prop}
	\begin{proof}
		See Appendix C.
	\end{proof}

	Notice that the component $C_{(0,\mathbf{s},\mathbf{r})}$ is the cost coefficient corresponding to a born track hypothesis $(0,\mathbf{s},\mathbf{r})=(0,s_1,\cdots,s_N,r_1,\cdots,r_N)$. It makes sense to set the cost coefficient $C_{(0,\mathbf{s},\mathbf{r})}$ equal to $0$ when we do not allow any born targets during data association of the confirmed targets. In practical applications, we use the unassigned measurements to initialize new tracks \cite{Kuiren1990Issues}.

	\begin{rem}
		It is not required to check the condition $(\ref{Eq_3.41})$ in Algorithm $\ref{alg_2}$.
		When the condition $(\ref{Eq_3.41})$ is not satisfied, the optimal solution of the LP problem $(\ref{Eq_3.21})$ is a probabilistic solution and Algorithm $\ref{alg_2}$ updates the stacked state vector by the RCMKF.
		When the condition $(\ref{Eq_3.41})$ holds, the optimal solution of the LP problem $(\ref{Eq_3.21})$ is a $0$-$1$ assignment, and Algorithm $\ref{alg_2}$ uses the KF to update the stacked state vector with a lower computational complexity as stated in the following Corollary. 	
	\end{rem}
	
	\begin{cor}\label{coro_1}
		If the conditions of Proposition \ref{pro_2} hold, the computational complexity of RCMKF reduces to that of KF, i.e., $O(R_1\cdot m^3)$, where $R_1$ is the number of valid measurements, and $m$ is the dimension of the measurement.
		Moreover, if there are $\hat{R}_1$ measurements are shared with different targets, then the computational complexity of the RCMKF is $O((R_1 - \hat{R}_1)\cdot m^3+ (\hat{R}_1 \cdot m)^3$.
	\end{cor}
	\begin{proof}
		See Appendix D.
	\end{proof}
	
	\begin{rem}
		In numerical experiments, we found that the condition $(\ref{Eq_3.41})$ is frequently satisfied for targets with large spacing.
		Therefore, in practice, the computational complexity of the RCMKF in Algorithm $\ref{alg_2}$ usually does not reach the worst-case computational complexity of $O((R_1\cdot m)^3)$.
	\end{rem}
	
	Based on linear programing and the RCMKF, the RMM-MHT algorithm is summarized in Algorithm 1.
	\begin{algorithm}
		\caption{Randomized Multiple Model MHT}
		\label{alg_2}
				\begin{itemize}
     \item Require: $\{\mathbf{x}_k^{\tau}\}_{\tau=1}^{T}$, $\{\mathbf{Z}_{k+1},\cdots,\mathbf{Z}_{k+N}\}$;
     \item Ensure:	$\{\mathbf{x}_{k+1}^{\tau}\}_{\tau=1}^{T}$

    \begin{enumerate}
      \item establish the relaxed problem $(\ref{Eq_3.21})$;
      \item solve problem $(\ref{Eq_3.21})$ by the linear programming algorithm;
      \item obtain the probabilities of data associations $(\ref{Eq_3.22})$ and dynamic models $(\ref{Eq_3.27})$, respectively, and establish the stacked system $(\ref{Eq_3.31})$ and $(\ref{Eq_3.32})$;
      \item update the stacked state system $(\ref{Eq_3.31})$ and $(\ref{Eq_3.32})$ by RCMKF $(\ref{Eq_3.33})$-$(\ref{Eq_3.40})$ and obtain the state estimates $\{\mathbf{x}_{k+1}^{\tau}\}_{\tau=1}^{T}$.
      \item set $k = k+1$;
    \end{enumerate}
   \end{itemize}

	\end{algorithm}
	
	\subsection{Track management}
	Track initialization: There are a number of methods to initialize a new track, such as the single-point method \cite{Mallick2008Comparison}, and the two-point method \cite{bar2004estimation}.
	The key to track initialization is how to select the points to form a track. In the framework of RMM-MHT, we consider the initialization problem as follows:
	\begin{align}
	\nonumber \text{Minimize}~
	&\sum_{\mathbf{r}}C_{\mathbf{r}} P_{\mathbf{r}}\\
	\nonumber
	\text{Subject to:}\\
	\nonumber
	&\sum_{\mathbf{r}\setminus\{r_n\}}P_{\mathbf{r}}=1,\\
	\nonumber
	&\text{for}~n=1,\cdots,N~\text{and}~r_n=1,\cdots,R_n,\\
	& ~0\leq P_{\mathbf{r}} \leq 1~\text{for all}~\mathbf{r},
	\end{align}
	where $ \mathbf{r}=(r_1,\cdots,r_N) $ is the hypothesis of a new track, $ P_{\mathbf{r}} $ is the probability that $ \mathbf{r} $ forms a new track, and $ C_{\mathbf{r}}=-L_{\mathbf{r}} $ is the  cost coefficient, which is defined in Appendix A.
	When the probability of a local track hypothesis is $ 1 $ or greater than a threshold, we apply a track initialization method to initialize a track.
		
	Track termination: For a specific track $ \tau $, we consider it to lose a measurement in the current scan if its association probability $ P_{\tau,0}^h=1 $. Moreover, the target $ \tau $ is considered to be terminated if it loses measurements over more than $ N_l $ scans, where $ N_l $ is the threshold for track termination \cite{Mallick2012MHT}.
	
	\section{Simulation results}\label{sec_4}
	In this section, we consider a typical MMTT scenario where we include both closely spaced targets and maneuvering targets described in \cite{bar2004estimation}.
	The true trajectories of the three targets are shown in Fig. \ref{fig_5}.
	As shown in Fig. \ref{fig_5}, the three targets fly westward with their motion described by the constant velocity (CV) model, before executing a $ 1^\circ/s $ coordinated turn based on the constant turn (CT) model. Then the three targets fly southward based on the CV model, followed by a $ 3^\circ/s $ coordinated turn based on the CT model. Finally, they continue to fly westward based on the CV model after the turn.
	Here, the state transition matrices of the CV and CT model are
	\begin{equation}\nonumber
		F_{CV}=\left[
		\begin{array}{cccc}
			1 & \Delta T & 0 & 0\\
			0 & 1 & 0 & 0\\
			0 & 0 & 1 & \Delta T\\
			0 & 0 & 0 & 1
		\end{array}
		\right],
	\end{equation}
	and
	\begin{align}
	\nonumber F_{CT}=\left[
	\begin{array}{cccc}
	1 & \frac{\sin\omega \Delta T}{\omega} & 0 & -\frac{1-\cos\omega \Delta T}{\omega}\\
	0 & \cos\omega \Delta T & 0 & -\sin\omega \Delta T\\
	0 & \frac{1-\cos\omega \Delta T}{\omega} & 1 & \frac{\sin\omega \Delta T}{\omega}\\
	0 & \sin\omega \Delta T & 0 & \cos\omega \Delta T
	\end{array}
	\right],
	\end{align}
	respectively, where $ \omega $ is the angle velocity.
	The corresponding process noise is
	\begin{equation}\label{cov}
		Q=q*\left[
		\begin{array}{cccc}
		\Delta T^3 & \Delta T^2/2 & 0 & 0\\
		\Delta T^2/2 & \Delta T & 0 & 0\\
		0 & 0 & \Delta T^3/3 & \Delta T^2/2\\
		0 & 0 & \Delta T^2/2 & \Delta T
		\end{array}
		\right].
	\end{equation}
	In this scenario, there are two coordinated turns, i.e., $ \omega^{(1)}=1^\circ/s $ and $ \omega^{(2)}=-3^\circ/s $.
	
	\begin{figure}[ht]
		\centering
		\includegraphics[width=\hsize]{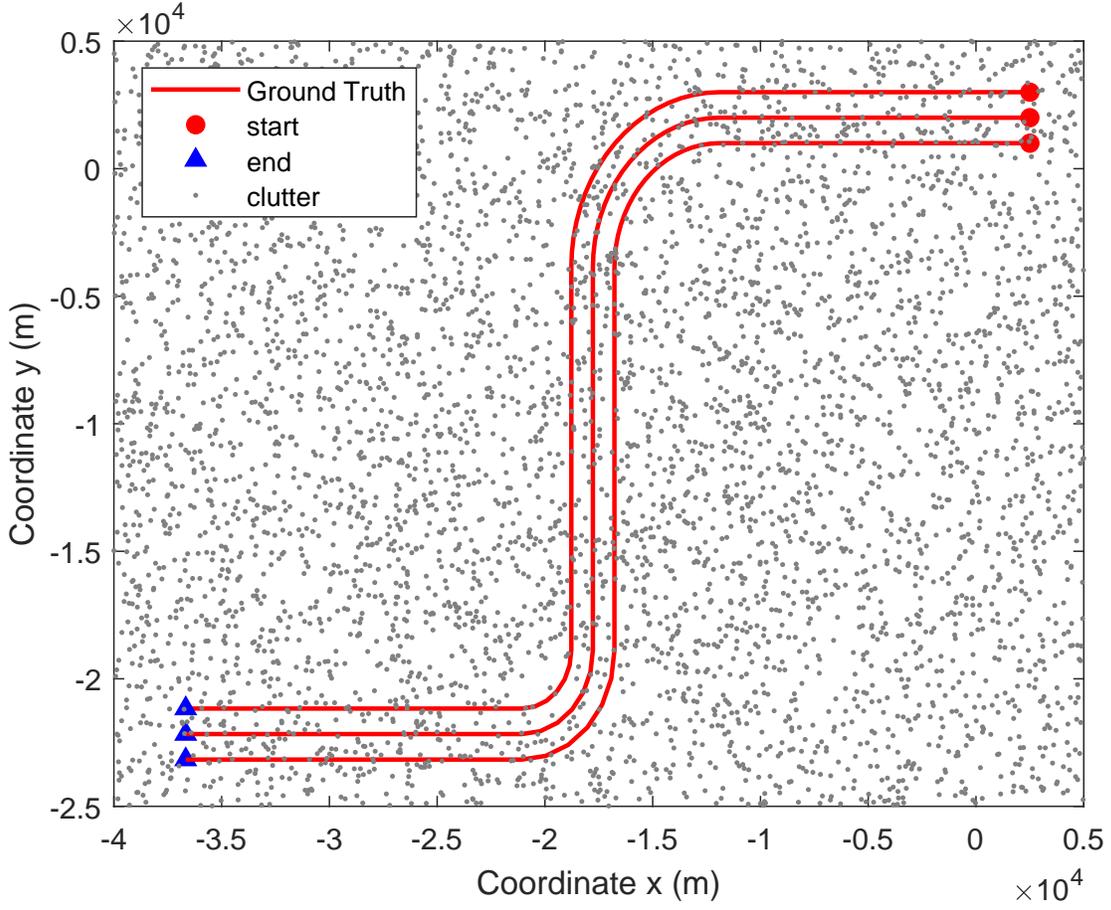}
		\caption{Scenario }
		\label{fig_5}
	\end{figure}
	
	Before we perform the comparison, a set of models is required for maneuvering target tracking.
	In general, the specific parameters of the turning speed are not known, but it is reasonable to use a large variance in the CV model to represent target maneuvers. Therefore, a set of CV models with two different variances is chosen for simulation of the above algorithms. Specifically, the parameters of the process noises of the two CV models are set to $q^{(1)}=0.01\text{m/s}^2$ and $q^{(2)}=4\text{m/s}^2$ in (\ref{cov}), respectively. The sampling period is $\Delta T=5$s.
	
	The simulation parameters are set as follows:
	the measurement model is the target position plus Gaussian noise of zero mean and standard deviation $ 400 \text{m}$.
	The clutter model assumes that the clutter returns originate from points that are uniformly distributed over the surveillance region,
	while the number of clutter returns is assumed to follow the Poisson distribution with a known mean $\lambda_{f}=50$ during each scan (see Fig. \ref{fig_5}).
	The target space is set to $1000$m.
	For each track, the two-point method \cite{bar2004estimation} is applied for track initialization and the tracking gate is performed with a threshold such that the probability of excluding the target-derived measurement is $10^{-4}$ \cite{Meyer2018}.

	\begin{align}
		\nonumber
		P^{(1)}=\left[
		\begin{array}{cc}
		0.95 & 0.05 \\
		0.1 & 0.9
		\end{array}
		\right],\\
		\nonumber
		P^{(2)}=\left[
		\begin{array}{cc}
		0.95 & 0.05 \\
		0.2 & 0.8
		\end{array}
		\right],\\
		\nonumber
		P^{(3)}=\left[
		\begin{array}{cc}
		0.95 & 0.05 \\
		0.3 & 0.7
		\end{array}
		\right].
	\end{align}	
	The corresponding IMM and IMM-MHT are denoted by IMM-P1, IMM-P2 and IMM-P3, IMM-MHT-P1, IMM-MHT-P2 and IMM-MHT-P3, respectively.
	The results of the simulation for RMM-MHT and IMM-MHT are compared using the tracking performance index: the average of the OSPA distance $ d_p^c(X,Y) $ \cite{SchuhmacherA} of $100$ Monte Carlo simulation runs.
	In Fig. \ref{fig_6}, the OSPA distances of the IMM filter based on different mode transition probability matrices are plotted as a function of the time steps, where we assume that the data association is correct. In Fig. \ref{fig_7}, the OSPA distances of the IMM-MHT and RMM-MHT, where we need to deal with the data association issue, are plotted as a function of the time steps, respectively.

	From Figs. \ref{fig_6}-\ref{fig_7}, we have the following observations:
	\begin{enumerate}
		\item Fig. \ref{fig_6} shows that, for this scenario, if there are no data association problems, the overall performance of the IMM filter is not sensitive to the mode transition probability matrix, except during the maneuvering stage, which is consistent with \cite{bar2004estimation,li2005survey}.
		\item Fig. \ref{fig_7} shows that when there are data association problems, the overall performance of IMM-MHT is sensitive to the mode transition probability matrix. Since if the IMM filter has a worse performance during the maneuvering stage like in Fig. \ref{fig_6}, it would lead to an incorrect data association so that the overall performance becomes worse.
		\item Fig. \ref{fig_7} also shows that the overall performance of RMM-MHT is close to the best performing IMM variant namely IMM-MHT-P1. Note that RMM-MHT does not require prior mode transition probabilities. Thus, it is more robust for tracking multiple maneuvering targets. The reason is that the probabilities of the target follow a specific dynamic model and the data association are derived by jointly optimizing the multiple possible models and data association hypotheses, whereas, IMM-MHT deals with them separately.
	\end{enumerate}

	\begin{figure}[ht]
		\centering
		\includegraphics[width=\hsize]{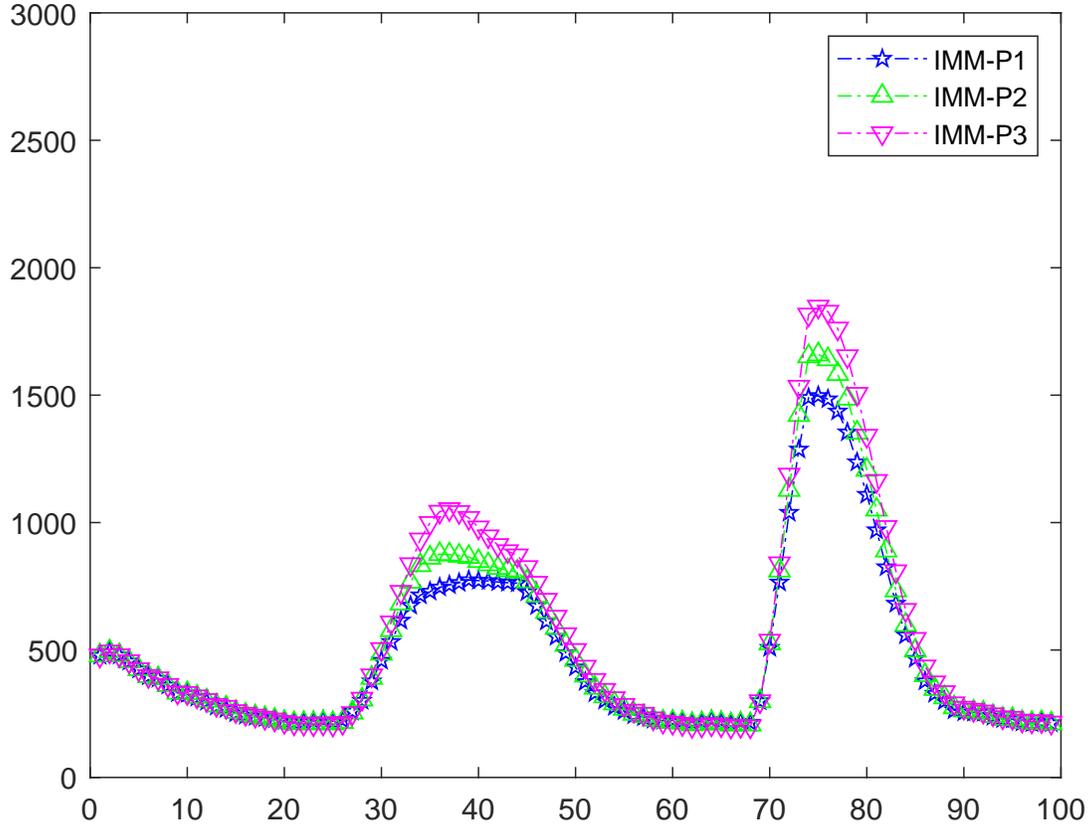}
		\caption{OSPA distances for different mode transition probability matrices}
		\label{fig_6}
	\end{figure}
	\begin{figure}[ht]
		\centering
		\includegraphics[width=\hsize]{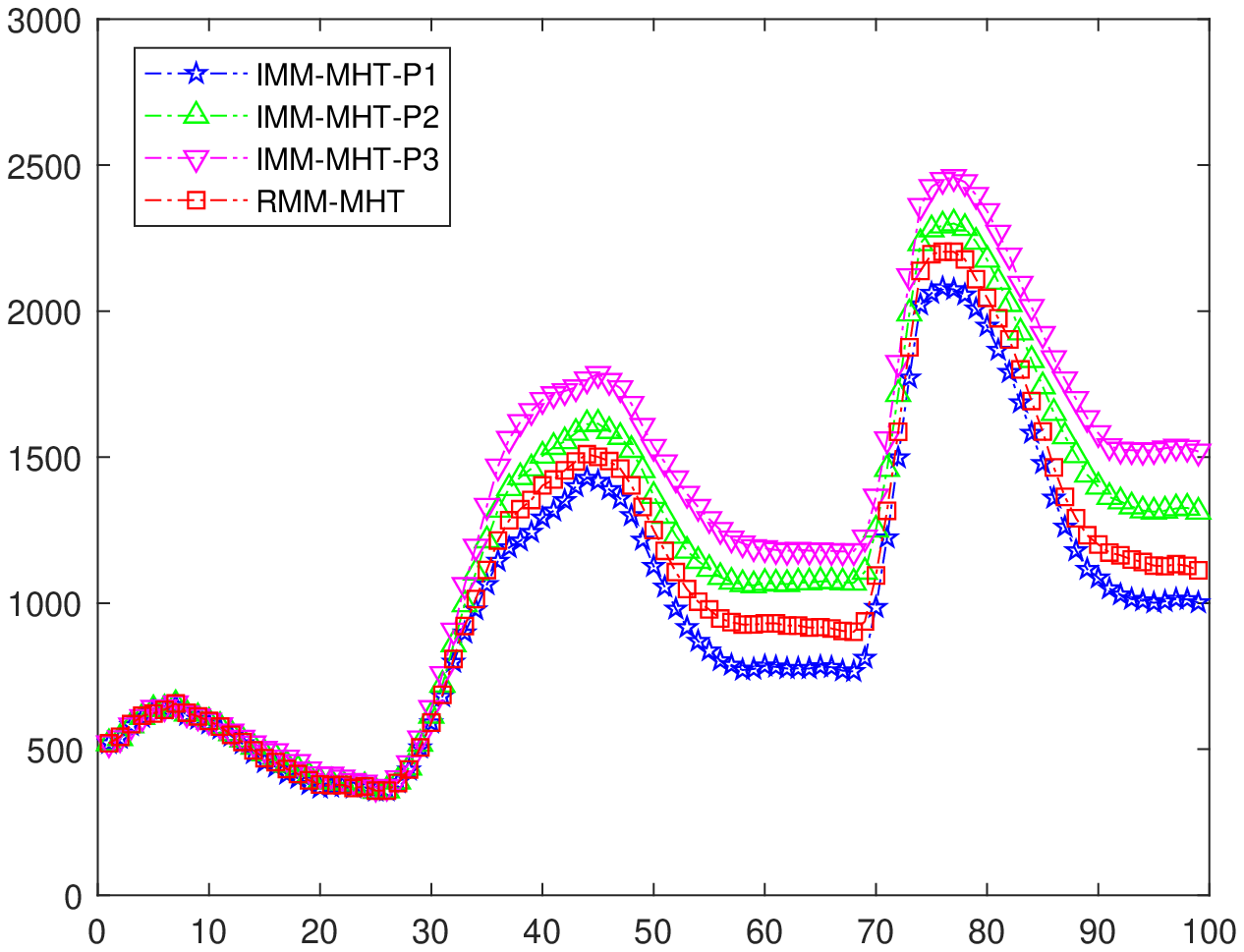}
		\caption{OSPA distances for IMM-MHT and RMM-MHT}
		\label{fig_7}
	\end{figure}
	
	\section{Conclusion}\label{sec_5}
	This paper proposed a novel RMM-MHT method for multi-target tracking. The main purpose of the proposed algorithm was to improve the tracking performance while enhancing robustness of MMTT with respect to the mode transition probabilities.
	In the proposed RMM-MHT algorithm, randomized data association avoids the NP-hard problem inherent in MHT, which reduces the probability of wrong associations.
	Since randomized data association introduces randomness in both the state and measurement equations, we generalized the RCMKF to improve the performance of state estimation.
	Moreover, RMM-MHT is suitable for tracking multiple maneuvering targets.
	It obtains the probability of the dynamic model of targets by solving the LP problem, which only requires a given set of models and does not assume a prior mode transition probability. Simulation results show that the RMM-MHT and IMM-MHT have a comparable performance and the RMM-MHT is robust since it does not require a prior knowledge of state transition probabilities.
	The future works include extending RMM-MHT to the case of nonlinear dynamic systems and multiple sensor networks.
	
\section*{Acknowledgements}
The authors would like to thank Yunmin Zhu, Yi Zhang, Yiwei Liao, Hanning Tang for helpful suggestions.

\appendix
\section{The calculation of Likelihood}
For each local hypothesis $ (\tau,\mathbf{s},\mathbf{r}) $, we reformulate the cost coefficient as follows:
\begin{align}\nonumber
C_{(\tau,\mathbf{s},\mathbf{r})}=-\log L_{(\tau,\mathbf{s},\mathbf{r})},
\end{align}
where $ L_{(\tau,\mathbf{s},\mathbf{r})} $ is the likelihood of the local hypothesis $ (\tau,\mathbf{s},\mathbf{r}) $.
Specifically, we define the likelihood $ L_{(\tau,\mathbf{s},\mathbf{r})} $ as
\begin{align}\label{appen_A}
\nonumber
L_{(\tau,\mathbf{s},\mathbf{r})}=\prod_{n=1}^{N}&(P_{\phi}^n)^{\Delta_{r_n}}
\cdot \left\{\left[\frac{P_d p_t(\mathbf{z}_n^{r_n}|\mathbf{x}^{\tau}, \mathbf{s}, \mathbf{r})}
{\lambda_f p_f(\mathbf{z}_n^{r_n})}
\right]^{\delta_{r_n}^n}\right.  \\
&\left. \cdot \left[\frac{\lambda_v  p_v(\mathbf{z}_n^{r_n})}
{\lambda_f p_f(\mathbf{z}_n^{r_n})}
\right]^{v_{r_n}^n}
\right\}^{1-\Delta_{r_n}},
\end{align}
where
\begin{align}
\nonumber
&\Delta_{r_n} = \left\{
\begin{array}{rcl}
1& & r_n = 0, \\
0& & \text{otherwise},
\end{array}\right.\\
\nonumber
&P^n_{\phi} = \left\{
\begin{array}{rcl}
&1-P_d  & \mathbf{z}_n^{r_n}~\text{is a detection}, \\
&1 &  \text{otherwise},
\end{array}\right.\\
\nonumber
&\delta_{r_n}^n = \left\{
\begin{array}{rcl}
&1  & \mathbf{z}_n^{r_n}~\text{belongs to an existing track}, \\
&0 &  \text{otherwise},
\end{array}\right.\\
\nonumber
&v_{r_n}^n = \left\{
\begin{array}{rcl}
&1  & \mathbf{z}_n^{r_n}~\text{initiates a track}, \\
&0 &  \text{otherwise},
\end{array}\right.
\end{align}
$ P_d $ is the probability of detection, $ \lambda_f $ is the expected number of clutter returns, $ \lambda_v $ is the expected number of born tracks. The corresponding densities in (\ref{appen_A}) are
\begin{align}
\nonumber
&p_f = p_v = 1/V,\\
\nonumber
&p_t(\mathbf{z}_n^{r_n}|\mathbf{x}^{\tau}, \mathbf{s}, \mathbf{r})=\mathcal{N}(\mathbf{z}_n^{r_n};\hat{\mathbf{z}}_n,\mathbf{S}_n),
\end{align}
where $ V $ is the volume of the surveillance region,
$ \hat{\mathbf{z}}_n $ and $ \mathbf{S}_n $ are the pseudo measurement and the corresponding covariance of $ \mathbf{x}^{\tau} $ at the $ n $-th scan, i.e.,
\begin{align}
\nonumber
&\hat{\mathbf{z}}_n = HF^{(s_n)}\cdots F^{(s_2)} F^{(s_1)}\mathbf{x}^{\tau},\\
\nonumber
&\mathbf{S}_n = H\{F^{(s_n)}\cdots  [F^{(s_1)}\mathbf{P}(F^{(s_1)})'+\mathbf{Q}^{(s_1)}]\\
\nonumber
&\qquad\qquad\qquad\times \cdots (F^{(s_n)})' + \mathbf{Q}^{(s_n)}\}H' + \mathbf{R}.
\end{align}

\section{Proof of Proposition 1}

We just need to prove that the noise processes $\{\tilde{\mathbf{v}}_k\}$ and $\{\tilde{\mathbf{w}}_k\}$ in system $(\ref{Eq_2.15})$ and $(\ref{Eq_2.16})$ satisfy the conditions of the standard Kalman Filter, i.e.,
\begin{align}
 \nonumber & E(\tilde{\mathbf{v}}_k)=0,\qquad E(\tilde{\mathbf{v}}_k\tilde{\mathbf{v}}'_l) =\widetilde{\mathbf{Q}}_k\delta_{k,l},\\
 \nonumber & E(\tilde{\mathbf{w}}_k)=0,\qquad E(\tilde{\mathbf{w}}_k\tilde{\mathbf{w}}'_l) =\widetilde{\mathbf{R}}_k\delta_{k,l},\\
 \nonumber & E(\tilde{\mathbf{v}}_k\tilde{\mathbf{w}}'_l)=0,\\
 \nonumber & E(\mathbf{x}_0\tilde{\mathbf{v}}'_k)=0,\qquad E(\mathbf{x}_0\tilde{\mathbf{w}}'_k)=0,
\end{align}
where
\begin{align}
   \nonumber
   \widetilde{\mathbf{Q}}_k &= \sum_{s=1}^S P_s^f (\mathbf{Q}^{(s)}_k+(F^{(s)}_k-\overline{\mathbf{F}}) E(\mathbf{x}_k\mathbf{x}'_k) (F^{(s)}_k-\overline{\mathbf{F}})'),\\
   \nonumber
   \widetilde{\mathbf{R}}_{k} &= \sum_{m=1}^M P_m^h (\mathbf{R}^{(m)}_{k}+(H^{(m)}_{k}-\overline{\mathbf{H}}) E(\mathbf{x}_{k}\mathbf{x}'_{k}) (H^{(m)}_{k}-\overline{\mathbf{H}})').
\end{align}

From the definition of $\tilde{\mathbf{v}}_k$ and condition 1), we have
\begin{align}
  \nonumber E(\tilde{\mathbf{v}}_k) =E(\mathbf{v}_k)+E(\widetilde{\mathbf{F}}_k)E(\mathbf{x}_k)=0.
\end{align}
By condition 1), $\mathbf{x}_0$ is independent of $\tilde{\mathbf{v}}_k$, thus
\begin{align}
  \nonumber E(\mathbf{x}_0\tilde{\mathbf{v}}'_k)=E(\mathbf{x}_0)E(\tilde{\mathbf{v}}'_k)=0
\end{align}

Without loss of generality, we assume that $k>l$, and from the definition of $\tilde{\mathbf{v}}_k$, we deduce that
\begin{align}
  \nonumber E(\tilde{\mathbf{v}}_k\tilde{\mathbf{v}}'_l)
  =E(\mathbf{v}_k\mathbf{v}'_l + \mathbf{v}_k\mathbf{x}'_l\widetilde{\mathbf{F}}'_l + \widetilde{\mathbf{F}}_k\mathbf{x}_k\mathbf{v}'_l + \widetilde{\mathbf{F}}_k\mathbf{x}_k\mathbf{x}'_l\widetilde{\mathbf{F}}'_l).
\end{align}
By conditions 1) and 2), we have $E(\mathbf{v}_k\mathbf{v}'_l)=0$, $E(\mathbf{v}_k\mathbf{x}'_l\widetilde{\mathbf{F}}'_l)=0$, $E(\widetilde{\mathbf{F}}_k\mathbf{x}_k\mathbf{v}'_l)=0$, and  $E(\widetilde{\mathbf{F}}_k\mathbf{x}_k\mathbf{x}'_l\widetilde{\mathbf{F}}'_l)=0$. Thus, $E(\tilde{\mathbf{v}}_k\tilde{\mathbf{v}}'_l)=0$ for $k\neq l$. Moreover, when $k=l$,
\begin{align}
\nonumber  &E(\tilde{\mathbf{v}}_k\tilde{\mathbf{v}}_k) \\
\nonumber
&= E(\mathbf{v}_k\mathbf{v}'_k + \mathbf{v}_k\mathbf{x}'_k\widetilde{\mathbf{F}}'_k + \widetilde{\mathbf{F}}_k\mathbf{x}_k\mathbf{v}'_k + \widetilde{\mathbf{F}}_k\mathbf{x}_k\mathbf{x}'_l\widetilde{\mathbf{F}}'_k)\\
\nonumber
&= \sum_{s=1}^{S}P_s^f \{E(v^{(s)}_k(v^{(s)}_k)')+E(v^{(s)}_k)\mathbf{x}'_k(F^{(s)}_k-\overline{\mathbf{F}})'\\
\nonumber
&\qquad +(F^{(s)}_k-\overline{\mathbf{F}})\mathbf{x}_kE(v^{(s)}_k)'\\
\nonumber
&\qquad +(F^{(s)}_k-\overline{\mathbf{F}})\mathbf{x}_k\mathbf{x}'_k(F^{(s)}_k-\overline{\mathbf{F}})'\}\\
\nonumber
&= \sum_{s=1}^{S}P_s^f \{\mathbf{Q}^{(s)}_k + (F^{(s)}_k-\overline{\mathbf{F}})\mathbf{x}_k\mathbf{x}'_k(F^{(s)}_k-\overline{\mathbf{F}})'\}.
\end{align}
That is, $E(\tilde{\mathbf{v}}_k\tilde{\mathbf{v}}'_l) =\widetilde{\mathbf{Q}}_k\delta_{k,l}$. We can prove the result for $\widetilde{\mathbf{w}}$ in a similar way.

\section{Proof of Proposition 2}    
Since the LP problem $(\ref{Eq_3.21})$ is a relaxed version of the $0$-$1$ IP problem $(\ref{Eq_2.8})$, i.e., the optimal objective function value of problem $(\ref{Eq_3.21})$ is a lower bound on that of the problem $(\ref{Eq_2.8})$. From the assumption that $\delta$ is a feasible solution of the two problems, if we prove that $\delta$ is an optimal solution of the problem $(\ref{Eq_3.21})$, then $\delta$ is also an optimal solution of the problem $(\ref{Eq_2.8})$.
Thus, to prove the equivalence of the optimal values of the problems $(\ref{Eq_2.8})$ and $(\ref{Eq_3.21})$, we only need to prove that $\delta$ is an optimal solution of the problem $(\ref{Eq_3.21})$.

We start by proving that $\delta$ is a locally optimal solution of the problem $(\ref{Eq_3.21})$.
We say that $d$ is a feasible direction, if there is a positive scalar $\theta$, $\hat{\delta}=\delta+\theta d$ is a feasible solution of the problem $(\ref{Eq_3.21})$. For any feasible direction $d$ with some $d_{(\tau,\mathbf{s},\mathbf{r})}>0$, where $1\leq \tau \leq T$ and $(\tau,\mathbf{s},\mathbf{r})\neq (\tau,\mathbf{s}^{\tau},\mathbf{r}^{\tau})$,
we compare the objective function values of $\delta$ and $\hat{\delta}$,
\begin{eqnarray}
\nonumber C'\hat{\delta}-C'\delta=\theta C'd.
\end{eqnarray}
Here, we use the notation $(\mathbf{s}^{\tau},\mathbf{r}^{\tau})^c$ to shorten the expression of $(\mathbf{s},\mathbf{r}) \neq (\mathbf{s}^{\tau},\mathbf{r}^{\tau})$.
Since $\theta>0$, we conclude that $d$ is not a decent direction when $C'd>0$. For each $\tau \neq 0$, from the feasibility of $\delta$ and $\hat{\delta}$, we have
\begin{align}\label{Eq_A1}
\nonumber
\sum_{\mathbf{s}}\sum_{\mathbf{r}}\delta_{(\tau,\mathbf{s},\mathbf{r})}
&=\sum_{\mathbf{s}}\sum_{\mathbf{r}}\hat{\delta}_{(\tau,\mathbf{s},\mathbf{r})}\\
&=\sum_{\mathbf{s}}\sum_{\mathbf{r}}(\delta_{(\tau,\mathbf{s},\mathbf{r})} +\theta d_{(\tau,\mathbf{s},\mathbf{r})}).
\end{align}
Reformulating Equation (\ref{Eq_A1}), we have
\begin{equation}\label{Eq_A2}
d_{(\tau,\mathbf{s}^{\tau},\mathbf{r}^{\tau})}+\sum_{(\mathbf{s}^{\tau},\mathbf{r}^{\tau})^c}d_{(\tau,\mathbf{s},\mathbf{r})}=0.
\end{equation}
We are now in a position to show the local optimality of $\delta$, i.e.,
$C'd>0$. We first split $C'd$ into three parts:
\begin{align}
\nonumber C'd
\nonumber &=\sum_{\tau\neq 0} C_{(\tau,\mathbf{s}^{\tau},\mathbf{r}^{\tau})}d_{(\tau,\mathbf{s}^{\tau},\mathbf{r}^{\tau})}+\sum_{\tau\neq 0}\sum_{(\mathbf{s}^{\tau},\mathbf{r}^{\tau})^c}C_{(\tau,\mathbf{s},\mathbf{r})}d_{(\tau,\mathbf{s},\mathbf{r})}\\
\nonumber &~~ +\sum_{\mathbf{s}}\sum_{\mathbf{r}}C_{(0,\mathbf{s},\mathbf{r})}d_{(0,\mathbf{s},\mathbf{r})}.
\end{align}
By assumptions $(\ref{Eq_3.41})$ and $(\ref{Eq_3.42})$, we conclude that
\begin{align}
\nonumber C'd
&=\sum_{\tau\neq 0} \{ C_{(\tau,\mathbf{s}^{\tau},\mathbf{r}^{\tau})}d_{(\tau,\mathbf{s}^{\tau},\mathbf{r}^{\tau})} +\sum_{(\mathbf{s}^{\tau},\mathbf{r}^{\tau})^c}C_{(\tau,\mathbf{s},\mathbf{r})}d_{(\tau,\mathbf{s},\mathbf{r})} \}\\
\nonumber
&>\sum_{\tau\neq 0}\{ C_{(\tau,\mathbf{s}^{\tau},\mathbf{r}^{\tau})}
(d_{(\tau,\mathbf{s}^{\tau},\mathbf{r}^{\tau})}+\sum_{(\mathbf{s}^{\tau},\mathbf{r}^{\tau})^c} d_{(\tau,\mathbf{s},\mathbf{r})} ) \}\\
\nonumber
&=0,
\end{align}
where the last equality follows from $({\ref{Eq_A2}})$. The fact that $C'd>0$ means that $d$ is not a descent direction, which also implies that a possible descent direction $d$ should satisfy $d_{(0,\mathbf{s},\mathbf{r})}>0$. Actually, from $C_{(0,\mathbf{s},\mathbf{r})}=0$, a feasible direction $d$ with $d_{(0,\mathbf{s},\mathbf{r})}>0$ cannot change the objective function, i.e., $C_{(0,\mathbf{s},\mathbf{r})}d_{(0,\mathbf{s},\mathbf{r})}=0$. Thus, we conclude that $\delta$ is a locally optimal solution of problem $({\ref{Eq_3.21}})$.

Since a locally optimal solution is also a globally optimal solution in LP, the proof of the equivalence of the optimal values of the two problems is complete. It should be noted that $\delta$ is not a unique globally optimal solution of problem $({\ref{Eq_3.21}})$.

\section{Proof of Corollary 1}
Without loss of generality, we assume that all the targets are ``uncoupled'' at time step $k$, i.e.,
\begin{align}
    \nonumber&\mathbf{X}_{k} = \left((\mathbf{x}_{k}^1)',\cdots,(\mathbf{x}_{k}^T)'\right)',\\
    \nonumber&\mathbf{P}_{k|k} = \text{diag}(\mathbf{P}_{k|k}^1,\cdots,\mathbf{P}_{k|k}^T).
\end{align}
Since the number of valid measurements is $R_1$, and the number of valid measurements that are shared with different targets is $\hat{R}_1$, we conclude that there are $R_1 - \hat{R}_1$ valid measurements that are associated with targets with probability 1. Without loss of generality, we assume that the first $R_1 - \hat{R}_1$ targets that are associated to different measurements with probability $1$.
Since the most computationally intensive part of the RCMKF is the matrix inverse in $(\ref{Eq_3.39})$, we need to verify that the $(\bar{\mathbf{h}}_{k+1}\mathbf{P}_{k+1|k}\bar{\mathbf{h}}_{k+1}'+\widetilde{\mathbf{R}}_{k+1})$ is block diagonal with the following form,
\begin{equation}\label{Eq_A3}
\left(
\begin{array}{cccc}
A_1 &  &  & \\
    & \cdots &  & \\
    &  & A_{R_1-\hat{R}_1} & \\
    &  &  & B
\end{array}
\right),
\end{equation}
where $A_i$ is a square matrix with size $m$, for $i = 1,\cdots, R_1 - \hat{R}_1$, $B$ is a square matrix with size $(\hat{R}_1\cdot m)$.
In addition, it is easy to check that $\mathbf{R}_{k+1}$ and $E(\tilde{\mathbf{h}}_{k+1}E(\mathbf{X}_{k+1}\mathbf{X}_{k+1}')\tilde{\mathbf{h}}_{k+1}')$ are block diagonal each entry of which is a square matrix with size $m$.
Thus, if we prove that $\bar{\mathbf{h}}_{k+1}\mathbf{P}_{k+1|k}\bar{\mathbf{h}}_{k+1}'$ is in the form of $(\ref{Eq_A3})$, then the computational complexity is $O((R_1-\hat{R}_1)\cdot m^3 + (\hat{R}_1\cdot m)^3)$.

Now, we are in a position to prove that $\bar{\mathbf{h}}_{k+1}\mathbf{P}_{k+1|k}\bar{\mathbf{h}}_{k+1}'$ is in the form of $(\ref{Eq_A3})$.
From $(\ref{Eq_3.38})$ and the definition of $\mathbf{F}_k$, $\mathbf{P}_{k|k}$ and $\mathbf{Q}_k$, it is easy to verify that $\mathbf{P}_{k+1|k}$ is block diagonal. By the above assumptions, we have the analytical expression of $\bar{\mathbf{h}}_{k+1}$ as follows,
\begin{equation}\nonumber
\bar{\mathbf{h}}_{k+1}=\left(
\begin{array}{cccc}
H &  &  & \\
    & \cdots &  & \\
    &  & H & \\
    &  &  & \bar{\mathbf{h}}
\end{array}
\right),
\end{equation}
where $\bar{\mathbf{h}}$ is the mean of measurement matrix of the last $\hat{R}_1$ valid measurements. Thus, $\bar{\mathbf{h}}_{k+1}\mathbf{P}_{k+1|k}\bar{\mathbf{h}}_{k+1}'$ is in the form of (\ref{Eq_A3}), we conclude that the computational complexity is $O((R_1-\hat{R}_1)\cdot m^3 + (\hat{R}_1\cdot m)^3)$.
Moreover, if the conditions of Proposition \ref{pro_2} holds, i.e., $\hat{R}_1=0$, then the computational complexity is $O(R_1\cdot m^3)$.

\end{document}